\documentclass[journal]{IEEEtran}

\pdfoutput=1

\usepackage{amsmath} 
\usepackage{amssymb}  
\usepackage{amsfonts}  
\usepackage{amsthm}  

\usepackage{graphicx}

\usepackage{pdflscape}
\usepackage{rotating}
\usepackage{color}
\usepackage{epstopdf}

\usepackage[labelfont=bf]{caption}
\usepackage{subcaption}

\usepackage{hyperref}

\newcommand{\col}{\mbox{col}}

\def\hal{{1 \over 2}}

\def\Linf{{\cal L}_\infty}

\def\L2e{{\cal L}_{2e}}

\def\rea{\mathbb{R}}

\def\diag{\mbox{diag}}

\def\begequarr{\begin{eqnarray}}
\def\endequarr{\end{eqnarray}}
\def\begequarrs{\begin{eqnarray*}}
\def\endequarrs{\end{eqnarray*}}
\def\begarr{\begin{array}}
\def\endarr{\end{array}}
\def\begequ{\begin{equation}}
\def\endequ{\end{equation}}
\def\lab{\label}
\def\begdes{\begin{description}}
\def\enddes{\end{description}}
\def\begenu{\begin{enumerate}}
\def\begite{\begin{itemize}}
\def\endite{\end{itemize}}
\def\endenu{\end{enumerate}}

\def\lef[{\left[\begin{array}}
\def\rig]{\end{array}\right]}
\def\qed{\hfill$\Box \Box \Box$}
\def\begcen{\begin{center}}
\def\endcen{\end{center}}
\def\begrem{\begin{remark}\rm}
\def\endrem{\end{remark}}
\newcommand{\enorm}[1]{\left| #1 \right|}
\newcommand{\lpnorm}[1]{\left\lVert #1 \right\rVert_\infty}

\def\TAC{{\it IEEE Transactions of Automatic Control}}
\def\AUT{{\it Automatica}}

\def\SCL{{\it Systems and Control Letters}}

\def\IJACSP{{\it Int. J. on Adaptive Control and Signal Processing}}

\def\et{\epsilon_t}

\def\diag{\mbox{diag}}

\def\rea{\mathbb{R}}

\newtheorem{assumption}{Assumption}
\newtheorem{proposition}{Proposition}

\newtheorem{lemma}{Lemma}
\newtheorem{remark}{Remark}
\newtheorem{corollary}{Corollary}

\def\qmx#1{\begin{bmatrix}#1\end{bmatrix}}

\def\beal#1{\begin{align}{#1}\end{align}}
\def\beals#1{\begin{align*}{#1}\end{align*}}

\usepackage{color}

\graphicspath{{fig/}{model/}}

\begin{document}
%
\title {A Robust Adaptive Flux Observer for a Class of Electromechanical Systems}

\author{Anton~Pyrkin,~\IEEEmembership{Member,~IEEE,}
Alexey~Vedyakov,~\IEEEmembership{Member,~IEEE,}
Romeo~Ortega,~\IEEEmembership{Fellow Member,~IEEE,}
Alexey~Bobtsov,~\IEEEmembership{Senior Member,~IEEE,}

\thanks{A. Pyrkin, A. Vedyakov, A. Bobtsov are with the Department of Control Systems and Informatics, ITMO University, Kronverkskiy av. 49, Saint Petersburg, 197101, Russia, e-mail: a.pyrkin@gmail.com}%
\thanks{Romeo~Ortega is with the LSS-CentraleSupelec, 3, Rue Joliot-Curie, 91192 Gif– sur–Yvette, France, e-mail:~ortega@lss.supelec.fr}
\thanks{This work is supported by the Russian Science Foundation grant (project 17-79-20341).}}
 
\maketitle
\begin{abstract}
The problem of designing a flux observer for magnetic field electromechanical systems from noise corrupted measurements of currents and voltages is addressed in this paper. Imposing a constraint on the systems magnetic energy function, which allows us to construct an algebraic relation between fluxes and measured voltages and currents that is independent of the mechanical coordinates, we identify a class of systems for which a globally convergent adaptive observer can be designed. A new adaptive observer design technique that effectively exploits the aforementioned algebraic relation is proposed and successfully applied to the practically important examples of permanent magnet synchronous motors and magnetic levitation systems.
\end{abstract}

%
\section{Introduction}
\lab{sec1}
%
High performance regulation of electromechanical systems usually requires the knowledge of some state variables that are not easy to measure. For instance, in AC electrical machines the de-facto standard controller is the so-called ``flux-orientation strategy" \cite{LEO}, that relies on the availability of the magnetic flux stored in the inductors, which is difficult to measure. Another example is magnetic levitation (MagLev) systems where measurement of position of the levitated object is of paramount importance but existing position sensors have low reliability and their cost is extremely high \cite{SCHMAS}. Actually, because of cost and dependability issues, one of the most challenging problems in a large class of electromechanical systems is the development of sensorless (also called self-sensing) controllers for them, by which it is understood that the control algorithms are based only on measurements of currents and voltages \cite{NAM}.   

The control community has been very active in the development of observers and partial state-feedback controllers for electromechanical systems, with main emphasis in electrical machines. A large number of articles and monographs written by control specialists have appeared in the last few years---see, for instance, \cite{BERPRA,BOBetalaut,ORTbook,VERbook} and references therein---with some of these developments having penetrated engineering practice. State estimation and sensorless control for permanent magnet synchronous motors (PMSM)  has been considered in various publications, including \cite{BOBetalaut,MALetal,ORTetalcst,TOMVER} with some interesting experimental evidence reported in \cite{CHOetal,VERetal}. For the case of MagLev systems the problem has recently been addressed in \cite{BOBetalmaglev}, see also \cite{SCHMAS} for some interesting technology-based solutions. 

In all the publications above it assumed that  there is no noise present in the measurement of currents and voltages, which is an unrealistic assumption in most applications. In this paper we are interested in the problem of robust estimation of the flux in a class of electromagnetic systems (including PMSM and MagLev systems), with the qualifier robust added to mean that the measurements are corrupted by noise, represented by an unknown, constant, bias. This situation, which is very common in electromechanical systems, is particularly critical for the developments reported in \cite{BOBetalaut,BOBetalmaglev} where the parameter estimation-based observer (PEBO) design technique, recently reported in \cite{ORTetalscl}, is used. As explained in Subsection \ref{subsec23} below, the presence of a constant bias in the measurements renders inapplicable this technique---see also Remark \ref{rem3}.  To solve the robust flux observation problem we impose in this paper and additional assumption on the magnetic energy function of the system, which allows us to construct a quadratic, algebraic relation between the unknown fluxes and the measured voltages and currents that is independent of the mechanical coordinates. The main outcome of this assumption is that, after a series of algebraic and filtering operations, it is possible to generate a linear regression model for the unknown flux and some constant parameters related to the measurement bias terms. Equipped with this regression model we can apply standard estimation techniques to design observers for the flux.\footnote{See \cite{FORetal} where a similar ``identification-based" approach is pursued within the context of robust output regulation.} In its simpler version, we estimate only the flux and treat the unknown parameters as additive disturbances---ensuring in this case, ultimate boundedness of the observation error. Furthermore, we propose an adaptive observer that estimates the unknown parameters to ensure global convergence of the flux observation error. As expected, both results rely on the assumption that the corresponding regressor signals satisfy a persistency of excitation (PE) assumption \cite{LJU,SASBOD}.   

The use of an additional algebraic constraint for flux observation was first exploited for PMSM in \cite{ORTetalcst}, yielding a locally convergent design based on a gradient descent. The observer was later rendered global in \cite{MALetal} and adaptive in  \cite{BERPRA,ROMetal}. This technique was later combined with PEBO in  \cite{BOBetalaut,BOBetalmaglev} where the potential robustness problem mentioned above was observed. The starting point for the developments reported here is a proposal made in \cite{BERPRA} to differentiate the quadratic algebraic constraint to obtain a linear relation and then use filtering techniques to remove the derivative terms.  One of the motivations for the analysis carried-out in \cite{BERPRA} was to get a better understanding of PEBO, with one of the outcomes being the establishment of a relationship between the gradient descent-based approach and PEBO. Since in the gradient descent-based approach no open-loop integrations are used, the potential drift instability phenomenon is avoided.\\ 

The two main contributions of  this paper are the following.
\begite
\item[C1] Extension of the  ``differentiation plus filtering" technique proposed in \cite{BERPRA} for the noise-free PMSM example, to a broader class of electromechanical systems---that includes MagLev systems. The class is identified imposing some constraints on the systems electrical energy function.
\item[C2]  The inclusion of additional filtering and nonlinear operation steps to treat the case of noisy measurements.
\endite
  
The remainder of the paper is organized as follows. Section \ref{sec2} gives the problem formulation.  In Section \ref{sec3} we identify the class of systems considered in the paper and give a key technical lemma.  In Section \ref{sec4} we present the two robust flux observers that we apply to the PMSM and the Maglev systems in Section \ref{sec5}, with some simulation results given for the latter. The paper is wrapped-up with concluding remarks and future research directions in  Section \ref{sec6}.\\

\noindent {\bf Notation.} ${\bf I}_n$ is the $n \times n$ identity matrix, ${\bf 0}_{n \times s}$ is an $n \times s$ matrix of zeros and ${\bf 1}_{n}$ is an $n$-dimensional vector of ones.  For $b \in \rea^n$, $A \in \rea^{n \times n}$ and $x:\rea_+ \to \rea^n$ we denote the Euclidean norm as $|b|$, the induced matrix norm as $\|A\|$,  the minimum and maximum eigenvalues as $\lambda_{\tt min}\{A\}$ and $\lambda_{\tt max}\{A\}$, respectively, and the $\Linf$-norm as $\|x\|_\infty$. All mappings are assumed smooth. Given a function $f:  \rea^n \to \rea$, we define the differential operator $\nabla f(x):=\left(\frac{\displaystyle \partial f }{\displaystyle \partial x}\right)^\top $ and for functions $g:  \rea^n \times \rea^m\to \rea$ we define  $\nabla_y g(x,y):=\left(\frac{\displaystyle \partial g }{\displaystyle \partial y}\right)^\top $. 

\section{Problem Formulation}
\lab{sec2}
%
In this section we formulate the robust observation problem addressed in the paper. First, the mathematical model of the system that we consider is presented. Then, the robust observation scenario adopted in the paper is given. Finally, the motivation to consider this scenario is discussed.
\subsection{The class of electromechanical systems}
\lab{subsec21}
%
In this paper we consider multiport magnetic field electromechanical systems consisting of $n_{E}$ magnetic ports  and $n_M$ mechanical ports as defined in \cite{MEI}. The electrical port variables are $(v,i)$, where $v \in \rea^{n_{E}}$ are the voltages and $i \in \rea^{n_{E}}$ the currents. The mechanical port variables are  $(F,\dot q)$, where $F\in \rea^{n_M}$ are the mechanical forces of electrical origin and  $\dot q \in \rea^{n_M}$ the (rotational or translational) velocities of the movable mechanical elements.  There are also (electrical and mechanical) dissipation elements, $m$ external voltage sources through which electrical energy is supplied to the magnetic elements that will be denoted $u \in \rea^m$, with $m \leq n_E$, and $n_M$ external, constant, (load) forces in the mechanical part, which are denoted $F_L \in \rea^{n_M}$.   

The magnetic energy stored in the inductances is defined by the function $H_E(\lambda,q)$, where $\lambda \in \rea^{n_M}$ is the vector of flux linkages. The mechanical energy stored by the movable part inertias is given by the function
$$
H_M(q,p)  =  \hal p^\top  M^{-1}(q) p+ U(q),
$$
where $p \in \rea^{n_M}$ are the momenta of the masses, $M(q) \in \rea^{n_M \times n_M}$ is the positive-definite, inertia matrix and $U(q) \in \rea$ denotes the potential energy function. As will become clear below, the mechanical dynamics plays no role in the solution of the flux observation problem and is given here only to complete the mathematical model of the system. 

The constitutive relations of the elements are
\begequarrs
i & = & \nabla_\lambda H_E(\lambda,q)\\
F & = & - \nabla_q H_E(\lambda,q)\\
\dot q & = &  \nabla_p H_M(q,p),
\endequarrs 
where the minus sign in the force of electrical origin of the second equation reflects Newton's third law. The equations of motion of the system can be described in port-Hamiltonian form as
\begin{align}
\qmx{\dot \lambda \\ \dot q \\ \dot p} & =\qmx{-R & {\bf 0}_{n_E \times n_M} & {\bf 0}_{n_E \times n_M} \\ {\bf 0}_{n_M \times n_E} & {\bf 0}_{n_M \times n_M} & {\bf I}_{n_M} \\ {\bf 0}_{n_M \times n_E} & -{\bf I}_{n_M} & -R_M} \nabla H(\lambda,q,p) 
\nonumber \\
& \quad + \qmx{Bu \\ {\bf 0}_{n_M} \\ F_L}
\lab{sys}
\end{align}
where we have defined the systems total energy function
$$
H(\lambda,q,p):= H_E(\lambda,q)+H_M(q,p)
$$
and we assumed the presence of resistive elements in series with the inductances, with $R \in \rea^{n_E \times n_E}$ the positive semi-definite resistor matrix, and the $n_M \times n_M$, positive semi-definite matrix $R_M$ accounts for the Coulomb friction effects, $B \in \rea^{n_E \times m}$ is a constant, full rank, input matrix to the electrical ports and we have included the external voltages $u$ and mechanical forces $F_L$. 

We bring to the readers attention the important fact that the flux dynamics can be rewriten as
\begequ
\lab{dotlam}
\dot \lambda  =  -R i + B u.
\endequ
Hence, if the electrical port variables $(Bu,i)$ are measurable, the derivative of the flux is known.

The class of systems described by \eqref{sys} is quite large and contains, as particular cases, many electrical machines and magnetic levitation systems---see Section \ref{sec5}. The description of the mathematical model given above is quite succinct, the interested reader is referred to \cite{MEI,STRDUI} for more details on modelling of general electromechanical systems and to \cite{LIUetal,ORTbook} for the particular case of electrical machines.

\begrem
\lab{rem1}
A particular case of electrical energy function, that contains the PMSM and MagLev systems studied in Section \ref{sec5}, is given by 
\begequ
\lab{parhe}
 H_E(\lambda,q)  =  \hal [\lambda - \mu(q)]^\top  L^{-1}(q)  [\lambda - \mu(q)],
\endequ
where $L(q) \in \rea^{n_E \times n_E}$ is the positive-definite inductance matrix and the vector $\mu(q) \in \rea^{n_E}$ represents the flux linkages due to permanent magnets. For this case the electrical constitutive relation reduces to
\begequ
\lab{lami}
\lambda  =  L(q) i  + \mu(q).
\endequ
\endrem   

\subsection{Formulation of the robust flux observation problem}
\lab{subsec22}
%
Consider the electromechanical system \eqref{sys}, where the form\footnote{See Assumption \ref{ass1} for an exact description of the required prior knowledge of this function.}  of the energy function $H_E(\lambda,q)$ and the matrices $R$ and $B$ are known.  Assume the only signals available for measurement are the currents $i$ and the voltage sources $u$, which are corrupted by constant unknown bias  terms $\delta_i \in \rea^{n_E}$ and $\delta_u \in \rea^{m}$, respectively, that is
\begequarr
 \nonumber
 i_m & = &  i  + \delta_i\\
u_m & = & u +  \delta_u,
\lab{mea}
\endequarr
with $i_m$ and $u_m$ the actual measured signals. Design a dynamical system
\begequarrs
\nonumber
\dot{\chi} & = & G(\chi,i_m, u_m)\\
\hat \lambda & = & H(\chi,i_m,u_m)
\label{dynobs}
\endequarrs
with $\chi \in \rea^{n_\chi}$, such that $\chi$ is bounded and
$$
\lim_{t\to\infty} |\tilde \lambda(t)|=0,
$$
for all initial conditions $(\lambda(0),q(0),p(0),\chi(0))$, where $\tilde \lambda:=\hat \lambda - \lambda$. As usual in observer applications we assume that  $u$ and $F_L$ are such that the systems state trajectories are bounded.
   
\subsection{Motivation for the robust observation problem scenario}
\lab{subsec23}
%
The non-robust version of the flux observation problem of Subsection \ref{subsec22} has been solved for several practical examples using the PEBO design technique reported in \cite{ORTetalscl}---that is, when it is assumed that $\delta_i=\delta_u=0$.  Indeed, in \cite{BOBetalijacsp} an adaptive flux observer for PMSMs with unknown resistance and inductance, but with known velocity, is proposed. In \cite{BOBetalaut} a sensorless controller (with known electrical parameters) is reported. Finally, the open problem of sensorless control of MagLev systems is solved in \cite{BOBetalmaglev}. In all these examples the use of a PEBO was instrumental to solve the problem. The main drawback of PEBO is that it relies on an open-loop integration, which in the aforementioned examples, is of the form
$$
\dot \xi =  -R i + u.
$$
Simple integration shows that, under the assumption of ideal measurements of $i$ and $u$, the following relation is established
$$
\lambda(t)=\xi(t)+\theta_0,
$$
with $\theta_0:=\lambda(0)-\xi(0)$. The essence of PEBO is to treat $\theta_0$ as an unknown constant parameter that is estimated with some parameter estimation algorithm. The observed flux is then obtained via
$$
\hat \lambda=\xi+\hat \theta_0,
$$
with  $\hat \theta_0$ being the estimated parameter. Clearly, a consistent estimation of this parameter yields asymptotic reconstruction of the flux. 

Unfortunately,  the open-loop integration might be a problematic operation in practice. For instance, in the biased measurement scenario described above, it gives rise to unbounded signals, styming the direct application of the PEBO technique. To propose a solution to this problem is the main motivation of the present  work.

\begrem
\lab{rem2}
As shown in \cite{ORTetalscl} PEBO relies on the solution of a partial differential equation that transforms the system into a suitable cascaded form. In the present case this step is obviated since, as seen from \eqref{dotlam}, the derivative of the part of the state to be estimated, {\em e.g.}, $\lambda$, is measurable. However, to avoid the drift problem mentioned above---that arises in the presence of (constant) measurement disturbances---we need to exploit additional properties of the system to design the robust observer. This condition is articulated in the assumption of the next section.   
\endrem

\begrem
\lab{rem3}
It should be mentioned that, in spite of the potential problem of open-loop integration of noisy signals,  with some additional ``safety nets" similar to the ones used in PID and adaptive control, PEBO has proven effective in the solution of various physical systems problems---and its performance has been validated experimentally \cite{BOBetalaut,BOBetalmaglev,CHOetal,PYRetal}.
\endrem
%
\section{Key Assumption and Main Technical Lemma}
\lab{sec3}
%
As indicated in the introduction, to solve the robust flux observation problem of Subsection \ref{subsec22} it is necessary to impose and additional assumption on the electrical subsystem \eqref{sys}, which is given in this section. This assumption, pertains only to the electrical energy function $H_E(\lambda,q)$ and it allows us to construct a quadratic, algebraic relation between the unknown fluxes $\lambda$ and the measured voltages $u_m$ and currents $i_m$ that is independent of the mechanical coordinates $q$.  Equipped with this assumption it is possible to generate a linear regression model for the unknown flux and a constant parameter related to the bias term $\delta_i$ and $\delta_u$. 
\subsection{Identification of the admissible systems}
\lab{subsec31}
%
We are in position to identify the class of systems that are considered in the paper.

\begin{assumption}
\lab{ass1} \em
Consider the electromechanical system \eqref{sys}. The electrical energy function $H_E(\lambda,q)$ is such that there exists  three matrices $Q_i \in \rea^{n_E \times n_E},\;i=1,2,3,$ a vector $C \in \rea^{n_E}$ and a scalar $d \in \rea$, all of them known and constant, such that the following algebraic relation holds
\begin{align}
	\lab{w}
	0 & = \lambda^\top \,Q_1\,\lambda+\lambda^\top  Q_2 i + i^\top  Q_3 i + C^\top  i + d
	\nonumber \\
	& =: w(\lambda,i),
\end{align}
where we recall that $i$ is given by
$$
i  =  \nabla_\lambda H_E(\lambda,q).
$$ 
\qed
\end{assumption}
 
In the robustness scenario considered in Subsection \ref{subsec22} the currents $i$ are perturbed by a constant bias. Therefore, to use the algebraic constraint \eqref{w} in the design it is necessary to express it in terms of the measurable currents $i_m$ that, as expected, gives rise to the appearance of unknown constant parameters. This result is contained in the corollary below. To simplify the notation, we refer in the sequel to ``known functions of time", meaning by that the notation 
$$
(\cdot)(t):=(\cdot)(i_m(t)).
$$ 
When clear from the context the time argument is omitted.

\begin{corollary}\em
\lab{cor1}
The algebraic constraint \eqref{w}, evaluated at $i=i_m-\delta_i$, may be expressed in the form
\begin{align}
\lab{wim}
w(\lambda,i_m-\delta_i) & = \lambda^\top \,Q_1\,\lambda + \lambda^\top  [y_a(t)+ \theta_{y_a}]
\nonumber \\
& \quad + y_b^\top (t)\theta_{y_b}+y_c(t)+d,
\end{align}
where $y_a \in \rea^{n_E},\;y_b\in \rea^{n_E+1}$ and $y_c \in \rea$ are known functions of time given by
\begin{align}
\lab{yabc}
y_a &:=Q_2i_m \nonumber \\
y_b &:=\qmx{ i_m \\ 1} \\
y_c &:=i_m^\top Q_3 i_m+C^\top i_m, \nonumber
\end{align}
and $\theta_{y_a}\in \rea^{n_E},\;\theta_{y_b}\in \rea^{n_E+1}$ are unknown, constant vectors.
\end{corollary}

\begin{proof}
The proof proceeds via direct calculation with \eqref{yabc} and the following definition of the unknown parameters
\begequ
\lab{theyab}
\theta_{y_a}:=Q_2 \delta_i,\;\theta_{y_b}:=\qmx{2Q_3\delta_i  \\ \delta_i^\top Q_3 \delta_i+C^\top  \delta_i}.
\endequ
\end{proof}

\begrem
\lab{rem4}
It is hard to give a physical interpretation to Assumption \ref{ass1}. However, as shown in  Section \ref{sec5}, it turns out to be verified by PMSM and MagLev systems and, as recently proven in \cite{ORTBOB}, it also holds true for the saturated model of the switched reluctance motor. See also \cite{BOBetalaut,BOBetalmaglev} where the existence of a relationship of the form \eqref{w} is exploited in the context of PEBO without measurement noise.
\endrem
\subsection{Generation of the linear regression}
\lab{subsec32}
%
In this subsection we prove that for systems verifying Assumption \ref{ass1} it is possible to generate a linear regression model to which we apply standard estimation techniques to design (robust and adaptive) flux observers in Section \ref{sec4}. 

To streamline the presentation of the result we fix a constant $\nu>0$ and introduce the following filters
\beal{
\nonumber
\dot\xi_1&=-\nu\xi_1+\nu y_m\\
\nonumber\dot\xi_2&=-\nu\xi_2+2\nu Q_1y_m-\nu^2 y_a\\
\nonumber\dot\xi_3&=-\nu\xi_3+\nu y_b\\
\nonumber\dot\xi_4&=-\nu\xi_4+\xi_2+ 2Q_1 y_m\\
\nonumber
\dot\xi_5&=-\nu\xi_5+y_m^\top \xi_2+\nu^2 y_c\\
\nonumber\dot\xi_6&=-\nu\xi_6+\nu\xi_4-\xi_2\\
\nonumber\dot\xi_7&=-\nu\xi_7+\nu \xi_1\\
\nonumber\dot\xi_{8}&=-\nu\xi_{8}+\nu [y_b-\xi_3]\\
\lab{fil}
\dot\xi_{9}&=-\nu\xi_{9}+\nu \xi_5-\nu^2 y_c+y_m^\top [\nu\xi_4-\xi_2],
} 
that operate on the known signals
\begequ
\lab{ym}
y_m:=-Ri_m+Bu_m,
\endequ
and $y_a,y_b,y_c$ defined in \eqref{yabc}. Also, we define the signals
\beal{
\nonumber
y&:=\xi_5-\nu y_c-\xi_9 \in \rea\\
\nonumber
\Phi_\lambda&:= 2\xi_2+y_a-\nu\xi_4 \in \rea^{n_E}\\
\lab{yphiphi}
\Phi_\theta&:=\qmx{ 2\xi_6\\
\xi_1-\xi_7\\
\nu[y_b-\xi_3-\xi_8]\\
1
} \in \rea^{n_\theta},
}
with $n_\theta:=3n_E + 2$.

The lemma below is instrumental for the design of the robust flux observers. The proof, being technically involved, is given in Appendix A.

\begin{lemma}\em
\lab{lem1}
Consider the electromechanical system \eqref{sys} verifying Assumption \ref{ass1} with the measurable signals \eqref{mea}. The following linear regression model holds
\begin{align}
\nonumber
\dot \lambda &=-Ri_m+Bu_m + \theta_m \\
\lab{linear0}
y&=\Phi_\lambda^\top\,\lambda+\Phi_\theta^\top\,\theta+\et,
\end{align}
where the known functions  $y$, $\Phi_\lambda$, and $\Phi_\theta$ are defined via \eqref{yabc}-\eqref{yphiphi}, we introduced the two unknown, constant vectors
\begequ
\lab{the}
\theta_m:=R\delta_i-B\delta_u \in \rea^{n_E},\;\theta:=\qmx{\theta_m \\ \theta_{y_a} \\ \theta_{y_b} \\ {2 \over \nu}\theta_m^\top  Q_1\theta_m} \in \rea^{n_\theta},
\endequ
and $\et$ is an exponentially decaying signal.\footnote{In the sequel we will use the symbol $\et$ to (generically) denote signals bounded by exponentially decaying functions of time. When clear from the context the symbol will be omitted.}

\qed
\end{lemma}

\begrem
\lab{rem5}
The motivation for the inclusion of the filtered signals $\xi_j,j=1,..,9,$ is given in the proof of Lemma \ref{lem1}. We underscore the presence of nonlinear operations in the generation of $\xi_5$ and $\xi_9$. 
\endrem
%
\section{Robust Flux Observers}
\lab{sec4}
%
In this section we present the main results of the paper, namely, two flux observers that are robust {\em vis-\`a-vis} measurement noise, represented by the constant bias \eqref{mea}. The observers follow as a direct application of gradient design techniques to the linear regression model \eqref{linear0} and, as expected, their convergence properties rely on a PE assumption imposed to the regressor vector \cite{LJU,SASBOD}.     
\subsection{Non-adaptive flux observer}
\lab{subsec41}
%
In this subsection we give an observer where the constant parameters $\theta$ and $\theta_m$ are treated as perturbations and prove that the observation error enters, exponentially fast, a residual set whose size is of the order of the measurement noise.

\begin{proposition}\em
\lab{pro1}
Consider the electromechanical system \eqref{sys}  verifying Assumption \ref{ass1} and the measured signals \eqref{mea}.  Assume the vector $\Phi_\lambda$, defined in \eqref{yphiphi} is PE. That is, there exists positive numbers $T_r$ and $\alpha_r$ such that
\begequ
\lab{pe1}
	\int_t^{t+T_r} \Phi_\lambda(s)\Phi_\lambda^\top (s) ds \geq \alpha_r {\bf I}_{n_E}, \; \forall t \geq 0.
\endequ
The flux observer  \eqref{yabc}-\eqref{yphiphi} and
\begin{align}
\lab{lam_hat}
\dot{\hat \lambda} = -R i_m + B u_m + \Gamma_r \Phi_\lambda [y-\Phi_\lambda^\top \hat \lambda],
\end{align}
where $\Gamma_r \in \rea^{n_E \times n_E}$ is a positive definite, tuning gain matrix ensures that there exists positive constants $m_r$, $\rho_r$ and $\ell$ such that
\beal{
\lab{lam_til}
|\tilde \lambda(t)| \leq m_r e^{-\rho_r t}|\tilde \lambda(0)| +  \ell, \; \forall t \geq 0,
}
where
\begequ
\lab{bouell}
\ell \leq \kappa |\col\left(\delta_i,\delta_u,|\delta_i|^2,|\delta_u|^2 \right)|,
\endequ
for some positive constant $\kappa$ independent of the measurement bias terms $\delta_i,\delta_u$. 
\qed
\end{proposition}

\begin{proof} From \eqref{linear0} and \eqref{lam_hat}, neglecting the exponentially decaying term $\et$, we get:
\beal{
\nonumber
\dot{\tilde \lambda} & = \theta_m + \Gamma_r \Phi_\lambda [ y - \Phi_\lambda^\top \hat \lambda ]\\
\nonumber	& = \theta_m -\Gamma_r \Phi_\lambda [\Phi_\lambda^\top \tilde \lambda + \Phi_\theta^\top \theta ]\\
	& = - \Gamma_r \Phi_\lambda\Phi_\lambda^\top  \tilde \lambda + \theta_m  - \Gamma_r \Phi_\lambda\Phi_\theta^\top  \theta
\nonumber \\
	& = - \Gamma_r \Phi_\lambda\Phi_\lambda^\top  \tilde \lambda +\theta_m- \Gamma_r \Phi_\lambda \Phi_\theta^\top  \theta .
\lab{eq:rfo_psi_tilde}
}

Note that, since $i_m$ and $u_m$ are bounded, we have that $\Phi_\lambda$, $\Phi_\theta$ are also bounded. Therefore, the disturbance signal
\begequ
\lab{b}
 b:= \theta_m- \Gamma_r \Phi_\lambda \Phi_\theta^\top  \theta
\endequ
is bounded. Since $\Phi_\lambda$ is PE we can invoke Lemma 1 in~\cite{EFIFRA}  to get the bound
\begin{align}
	\label{eq:rfo_norm_bound}
	\enorm{\tilde{\lambda}(t)} \leq & {e^{\eta T_r} \over \eta}\left[\sqrt{\lambda_{\tt max}\{\Gamma_r\} }\lpnorm{\Phi_\lambda} \times \right. \nonumber \\
	& \left. \times e^{-{\eta \over 2} e^{-2\eta T_r t}} \enorm{\tilde{\lambda}(0)} + \lpnorm{b}\right],
\end{align}
where
$$
\eta := -{1 \over 2 T_r} \ln \left( 1 - \frac{\lambda_{\tt min}\{\Gamma_r\} \alpha_r}{ 1 + \lambda_{\tt max}\{\Gamma_r\}^2 T_r^2 \lpnorm{\Phi_\lambda}^4 } \right).
$$

The first part of the proof is completed by picking
\begin{align}
\kappa_r & :={e^{\eta T_r} \over \eta}\sqrt{\lambda_{\tt max}\{\Gamma_r\} }\lpnorm{\Phi_\lambda},
\nonumber \\
\rho_r & := \hal \eta e^{-2\eta T_r},\qquad \ell:={e^{\eta T_r} \over \eta}\lpnorm{b}.
\nonumber
\end{align}

To establish the bound \eqref{bouell} we apply the triangle inequality to \eqref{b}
\beals{
	\lpnorm{b} & \leq |\theta_m|+\lambda_{\tt max}\{\Gamma_r\}\lpnorm{\Phi_\lambda}\lpnorm{\Phi_\theta}|\theta |
	}
and, use the definitions of $\theta_m$ and $\theta$ given in \eqref{theyab} and \eqref{the}. 
\end{proof} 

\begrem
\lab{rem6}
A corollary of Proposition \ref{pro1} is that the observation error exponentially enters a residual set whose size is determined by the measurement error and, in the absence of it, the observer is exponentially convergent. Unfortunately, due to the complex operations that gave rise to the linear regression model \eqref{linear0}, determining the conditions under which the PE assumption \eqref{pe1} is satisfied is a daunting task. 
\endrem
  
\subsection{Adaptive flux observer}
\lab{subsec42}
%
In this subsection we give an adaptive version of the observer, where the unknown parameters $\theta_m$ and $\theta$ are estimated on-line. In this case, always under a PE assumption, exponential convergence to zero of the flux observation error is ensured.

\begin{proposition}\em
\lab{pro2}
Consider the electromechanical system \eqref{sys} verifying Assumption \ref{ass1} and the measured signals \eqref{mea}. Let the adaptive flux observer be given by  \eqref{yabc}-\eqref{yphiphi}. 
\begin{align}
\lab{psi_dot_ad}
\qmx{ \dot{\hat \lambda} \\ \dot{\hat\theta} } = & \qmx{ -R i_m + B u_m +G \hat\theta \\ 0}+
\nonumber \\
&  \gamma_a\,\Psi\,\left(y-\Phi_\lambda^\top \hat \lambda-\Phi_\theta^\top \hat\theta\right).
\end{align}
where  we defined 
\beals{
\Psi & :=\col(\Phi_\lambda,\Phi_\theta) \in \rea^{(4n_E +2)}\\
G&:=\qmx{{\bf I}_{n_E} & | & {\bf 0}_{n_E \times (2n_E+2)})} \in \rea^{n_E  \times (3 n_E+2)} ,
}
and  $\gamma_a \in \rea_+$ is a scalar tuning gain.\footnote{The choice of scalar adaptation gain is done to simplify the presentation.}  Assume $\Psi$ is PE, that is, there exists positive numbers $T_a$ and $\nu_a$ such that 
\begequ
\lab{pe2}
	\int_t^{t+T_a} \Psi(s) \Psi^\top (s)  ds \geq \nu_a {\bf I}_{(4n_E+2)}, \; \forall t \geq 0.
\endequ

There exists a constant $\gamma_{\tt min}>0$ such that for all $\gamma_a > \gamma_{\tt min}$ the inequality
$$
\left|\lef[{c} \tilde \lambda(t) \\ \tilde \theta(t)  \rig]\right|\leq m_a e^{-\rho_a t}\left|\lef[{c} \tilde \lambda(0) \\ \tilde \theta(0) \rig]\right|, \; \forall t \geq 0,
$$
holds for some positive constants $m_a$ and $\rho_a$ with ${\tilde \theta}:={\hat \theta} - \theta$.
\qed
\end{proposition}
\begin{proof} To simplify the notation define the vector $\chi:=\col( \tilde \lambda, \tilde \theta) \in \rea^{(4 n_E+2)}$. Using \eqref{linear0} and noticing that $G \theta=\theta_m$ the error model takes the form
\begin{align}
\label{eq:afo_sys}
\dot {\chi}  = -\gamma_a \Psi\Psi^\top   \chi + G_a\chi,
\end{align}
where $G_a \in \rea^{(4n_E +2) \times (4 n_E+2)}$:
$$
G_a:=\qmx{  {\bf 0}_{n_E \times n_E} &  G \\  {\bf 0}_{(3n_E +2) \times n_E} &  {\bf 0}_{(3n_E + 2) \times (3n_E+2)}}.
$$

Introduce the time-scaling 
$$
{dt \over d\tau} = {1 \over {\gamma_a}},
$$
and rewrite system~\eqref{eq:afo_sys}, in the $\tau$ time scale as 
\begin{align}
\label{eq:afo_sys_scaled}
{d\chi \over d\tau} = - \Psi(\tau) \Psi^\top (\tau)  \chi(\tau) +  {1 \over {\gamma_a}} G_a \chi(\tau).
\end{align}
Now, the PE property is invariant under time-scaling, hence $\Psi(\tau)$ is also PE. Consequently, the zero equilibrium of the unperturbed system 
$$
{d\chi \over d\tau} = - \Psi(\tau) \Psi^\top (\tau)  \chi(\tau)
$$
is GES \cite{AND,SASBOD}. Invoking Krasovskii's theorem~\cite{KRA} for GES systems there exists a continuously differentiable function $V(\chi,\tau)$ such that 
\begin{align}
	\nonumber
	& c_{1}\left\vert \chi\right\vert ^{2}\leq V(\tau,\chi)\leq c_{2}\left\vert \chi\right\vert ^{2}
	\\
	\nonumber
	& \frac{\partial V(\tau)}{\partial \tau}-\frac{\partial V}{\partial \chi}  \Psi(\tau) \Psi^\top (\tau)  \chi(\tau)\leq-c_{3}\left\vert \chi\right\vert ^{2}
	\\
	\label{dotv}
	& \left\vert \frac{\partial V}{\partial \chi}\right\vert \leq c_{4}\left\vert \chi\right\vert,
\end{align}
where $c_i,\;i=1,\dots,4$, are positive numbers. Using \eqref{dotv} we see that the time derivative of $V(\chi,\tau)$ along the solutions of \eqref{eq:afo_sys_scaled} satisfies the inequality
\begin{align}
	\dot V & \leq -\frac{1}{c_2} \left( c_3 - { c_4 \over \gamma_a}\|G_a\| \right) V.
	\nonumber
\end{align}
The proof is concluded setting
$$
	\gamma_{\tt min} :=\frac{c_4}{c_3} \|G_a\|.
$$
\end{proof}

\begrem
\lab{rem7}
Unfortunately, we cannot give an explicit expression for $\gamma_{\tt min}$ in terms of the PE parameters $T_a,\nu_a$ and $\|\Psi\|_\infty$. Indeed, with the existing estimates of convergence rate of systems of the form $\dot {\chi}=\Psi\Psi^\top \chi$, with $\Psi \in PE$,   it is not possible to accommodate a bounded additive disturbance of the form \eqref{eq:afo_sys_scaled}. However, for given numerical values of these constants, it is alway possible to solve the scalar inequality and numerically compute  $\gamma_{\tt min}$. 
\endrem
%
\section{Two Physical Examples}
\lab{sec5}
%
In this section we identify the parameters of the algebraic constraint \eqref{w} for two practically important physical examples:
\begite
\item[(i)] The PMSM extensively studied in the literature, see {\em e.g.}, \cite{BERPRA,BOBetalaut,VERetal} and references therein.
\item[(ii)] The benchmark MagLev system proposed in \cite{KNO} and considered in \cite{RODetal}.
\endite  
\subsection{Surface-mounted PMSM}
\lab{subsec51}
%
We consider the classical, two--phase $\alpha\beta$ model of the unsaturated, {non--salient} PMSM model given, for instance, in \cite{BOBetalaut}. In that case, the electrical energy function is given by \eqref{parhe} with  parameters 
\begin{align}
n_E & = 2,\;n_M=1,\;L(q)=L_s {\bf I}_2,\;
\nonumber \\
\mu(q) & = \lambda_m \left[ \begin{array}{c} \cos(n_p q) \\ \sin(n_p q) \end{array}\right],
\nonumber
\end{align}
where $L_s,\lambda_m,n_p$  are known positive constants. As it was first observed in \cite{ORTetalcst}, and later used in \cite{BERPRA,BOBetalaut,MALetal}, the following quadratic algebraic equation holds
$$
0=|\lambda-L_si|^2-\lambda_m^2.
$$
Developing this equation we get
\begequ
\lab{wpmsm}
0=|\lambda|^2-2 L_s\lambda^\top \,i+L_s^2|i|^2-\lambda_m^2.
\endequ
Equating the terms of \eqref{wpmsm} with those of \eqref{w} we identify the required constants as
\begin{align*}
Q_1 & = {\bf I}_2, & Q_2 & =-2L_s {\bf I}_2, & Q_3= L_s^2 {\bf I}_2 ,\;
\\
C & = 0, & d & =-\lambda_m^2. &
\end{align*}

Given the definition of these parameters it is straightforward to apply Lemma \ref{lem1} to construct the linear regression \eqref{linear0} and, using Propositions \ref{pro1} or \ref{pro2}, to design the robust flux observers---these additional steps are omitted for brevity.

\subsection{Two degrees of freedom MagLev system}
\lab{subsec52}
%

\begin{figure}
\label{fig1}
\centering
\includegraphics[width=\linewidth]{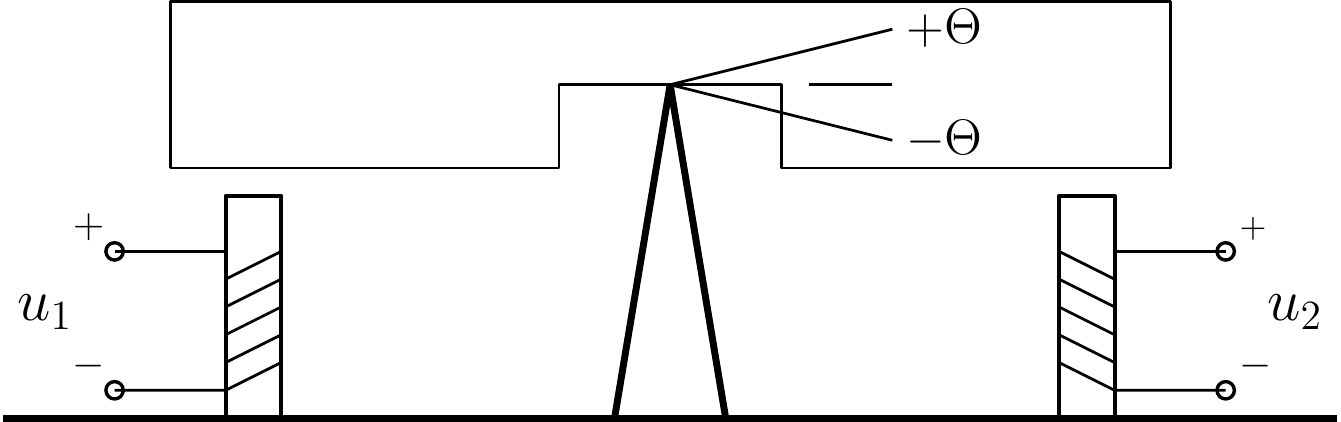}
\caption{$2$-DOF MagLev system \cite{KNO}}
\end{figure}

We consider in this Subsection the MagLev system depicted in Fig. 1, which was proposed as a benchmark example in \cite{KNO} and studied in \cite{RODetal}. To illustrate the performance of the proposed observer we give also simulation results of the system in closed-loop with the full-state feedback controller proposed in \cite{RODetal}. In order to make the presentation self-contained, we briefly present the systems model. The interested reader is refered to \cite{KNO,RODetal} for further details on the model and the control strategy.  

\subsubsection{Mathematical model}
\lab{subsubsec521}

The system has two electrical and one mechanical coordinate, therefore $n_E=2$ and $n_M=1$. The total energy function of the system is given by
\begin{align}
	H(\lambda,q,p)= \hal \lambda^\top  L^{-1}(q) \lambda + \frac{1}{2J}{p^2},
\lab{hmaglev}
\end{align}
with 
\begin{align}
	L(q) & =\qmx{{k_1\over g-q} & 0 \\ 0 & {k_2\over g+q}},
	\lab{parmaglev}
\end{align}
where $q=D\Theta$, $D$ is distance from pivot to actuator, which is denoted by $L$ in \cite{KNO} and \cite{RODetal}, $\Theta$ is rotation angle, and $k_1$, $k_2$, $g$ are known positive constants. The pH representation \eqref{sys} takes the form
\begin{align}
\lab{sysmaglev}
\qmx{\dot \lambda \\ \dot q \\ \dot p} & =\qmx{-{R \over Nc_1} & 0 & 0 & 0 \\ 0 & -{R \over Nc_2} & 0 & 0 \\ 0 & 0 & 0 & D \\ 0 & 0 & -D & 0 } \nabla H + \qmx{ u \\ {\bf 0}_{2 \times 1}}
\end{align}
where $u \in \rea^2$ are the control voltages applied to the inductors and $R$, $N$, $c_1$, and $c_2$ are known positive constants. Notice that 
\begequ
\lab{lamequli}
\lambda=L(q)i.
\endequ

We simulate the system in closed-loop with the nonlinear, static full-state feedback controller proposed in~\cite{RODetal} given as
\begin{align}
\nonumber
	u_1 & = R \frac{g - q}{k_1} \lambda_1 - \frac{R D}{\alpha c_1} \left(c_1\frac{\lambda_1^2}{2k_1} - c_2 \frac{\lambda_{2*}^2}{2k_2}\right) \\
\nonumber		
& \quad - \left(\frac{R D}{2\alpha c_1} + N\alpha\frac{R_a}{D}\right)G\tilde{z}_2 - N\frac{\alpha}{J}p,
	\\
\nonumber
u_2 & = R \frac{g + q}{k_2} \lambda_2 + \frac{R D}{\beta c_2} \left( c_2\frac{\lambda_2^2}{2k_2} - c_2 \frac{\lambda_{2*}^2}{2k_2} \right) \\
		& \quad - \left(\frac{R D}{2\beta c_2} + N\beta\frac{R_a}{D}\right)G \tilde{z}_2 - N\frac{\beta}{J}p,
\lab{idapbc}
\end{align}
where $\alpha > 0$, $\beta < 0$, $R_{a} > 0$, $G > 0$ are controller tuning parameters,
\begin{align*}
	\tilde{z}_2 & = \frac{D}{2 \alpha}(\lambda_1-\lambda_{1*}) + \frac{D}{2 \beta}(\lambda_2-\lambda_{2*})
	\\
	& \quad + D (q-q_{*}) + \frac{R_a}{D}(p-p_{*}),
\end{align*}
and $\lambda_{1*} = \sqrt{\frac{k_1 c_2}{k_2 c_1}} \lambda_{2*}$, $\lambda_{2*}$, $q_{*}$, $p_{*} = 0$ are the desired constant values for $\lambda_{1}$, $\lambda_{2}$, $q$, $p$ respectively.

\subsubsection{Regression form \eqref{w}}
\lab{subsubsec522}

From \eqref{lamequli} with \eqref{parmaglev} we easily get
\begequ
\lab{wmglv}
0=2g\lambda_1\lambda_2-k_1i_1\lambda_2-k_2i_2\lambda_1.
\endequ
Equating the terms of \eqref{wmglv} with those of \eqref{w} we identify 
\begin{align*}
	Q_1 & = \qmx{0 & g \\ g & 0}, & Q_2 & = \qmx{0 & -k_2 \\ -k_1 & 0}, & Q_3 &= {\bf 0}_{2 \times 2},
	\\
	C&=0, & d&=0, &\theta_{y_b}&=0.
\end{align*}

In this case, $y_b$ and $y_c$ are not used and $\xi_3$ and $\xi_8$ can be omitted. Then, the parameters of the regression form \eqref{linear0} become
\beal{
\nonumber
y&=\xi_5-\xi_9\\
\nonumber
\Phi_\lambda&= 2\xi_2+y_a-\nu\xi_4\\
\nonumber
\Phi_\theta&=\qmx{ 2\xi_6\\
\xi_1-\xi_7\\
0\\
1
},
}
and the following filters are used
\beal{
\nonumber
\dot\xi_1&=-\nu\xi_1+\nu y_m\\
\nonumber\dot\xi_2&=-\nu\xi_2+2\nu Q_1y_m-\nu^2 y_a\\
\nonumber\dot\xi_4&=-\nu\xi_4+\xi_2+ 2Q_1 y_m\\
\nonumber
\dot\xi_5&=-\nu\xi_5+y_m^\top \xi_2\\
\nonumber\dot\xi_6&=-\nu\xi_6+\nu\xi_4-\xi_2\\
\nonumber\dot\xi_7&=-\nu\xi_7+\nu \xi_1\\
\nonumber\dot\xi_{9}&=-\nu\xi_{9}+\nu \xi_5+y_m^\top [\nu\xi_4-\xi_2].
}
\subsubsection{Simulation results}
\lab{sec6}
%

The $2$-dof Maglev system \eqref{hmaglev}-\eqref{sysmaglev} in closed-loop with the controller \eqref{idapbc}
was simulated with the following plant parameters: 
\begin{gather*}
J = 9.67\times10^{-2},\;k_1 = k_2 = 2.2\times10^{-8},
\\
R = 1.6,\;N=321,\;c_1 = c_2 = 293.5,
\\
g = 3.3\times10^{-4},\;D=0.145.
\end{gather*}
In order to comply with the PE requirement we use  time-varying position reference, denoted $q_d(t)$:
$$
q_{d}(t) = 10^{-5} \left[ 6 \sin 10t + 4 \sin 20 t + 6 \sin 15t \right].
$$
For simplicity we set $\lambda_{1d}=\lambda_{2d}=p_d=0$.

The controller parameters were fixed at 
$$
\alpha=10,\;\beta=-10,\;R_a=10^{-2},
$$
which were tuned to reduce the overshoot. The filter parameter used in the observers was set at $\nu=50$ and the adaptation gains were selected as $\Gamma=10^4$ and $\gamma_a=\diag\{10^{20};10^{20};10^{20};10^{20};2\times10^{13};2\times10^{13};20\}$. 

In the simulations we set the measurement biases as follows
\begin{align}
\nonumber
\delta_i(t)&=\left\{\begin{array}{l}
\begin{bmatrix}
-0.003 & 0.0025
\end{bmatrix}^\top\!\!\!, \text{ for } 0 \le t \le 50 \text{ sec,}\\
\begin{bmatrix}
0.001 & 0.0008
\end{bmatrix}^\top\!\!\!, \text{ for } t \ge 50 \text{ sec,}
\end{array}
\right. \\
\delta_u(t)& =\left\{\begin{array}{l}
\begin{bmatrix}
0 & 0
\end{bmatrix}^\top\!\!\!, \text{ for } 0 \le t \le 50 \text{ sec,}\\
\begin{bmatrix}
0.002 & 0.0002
\end{bmatrix}^\top\!\!\!, \text{ for } t \ge 50 \text{ sec.}
\end{array}
\right.
\nonumber
\end{align}

\begin{figure}[htp]
\centering
\includegraphics[width=\linewidth]{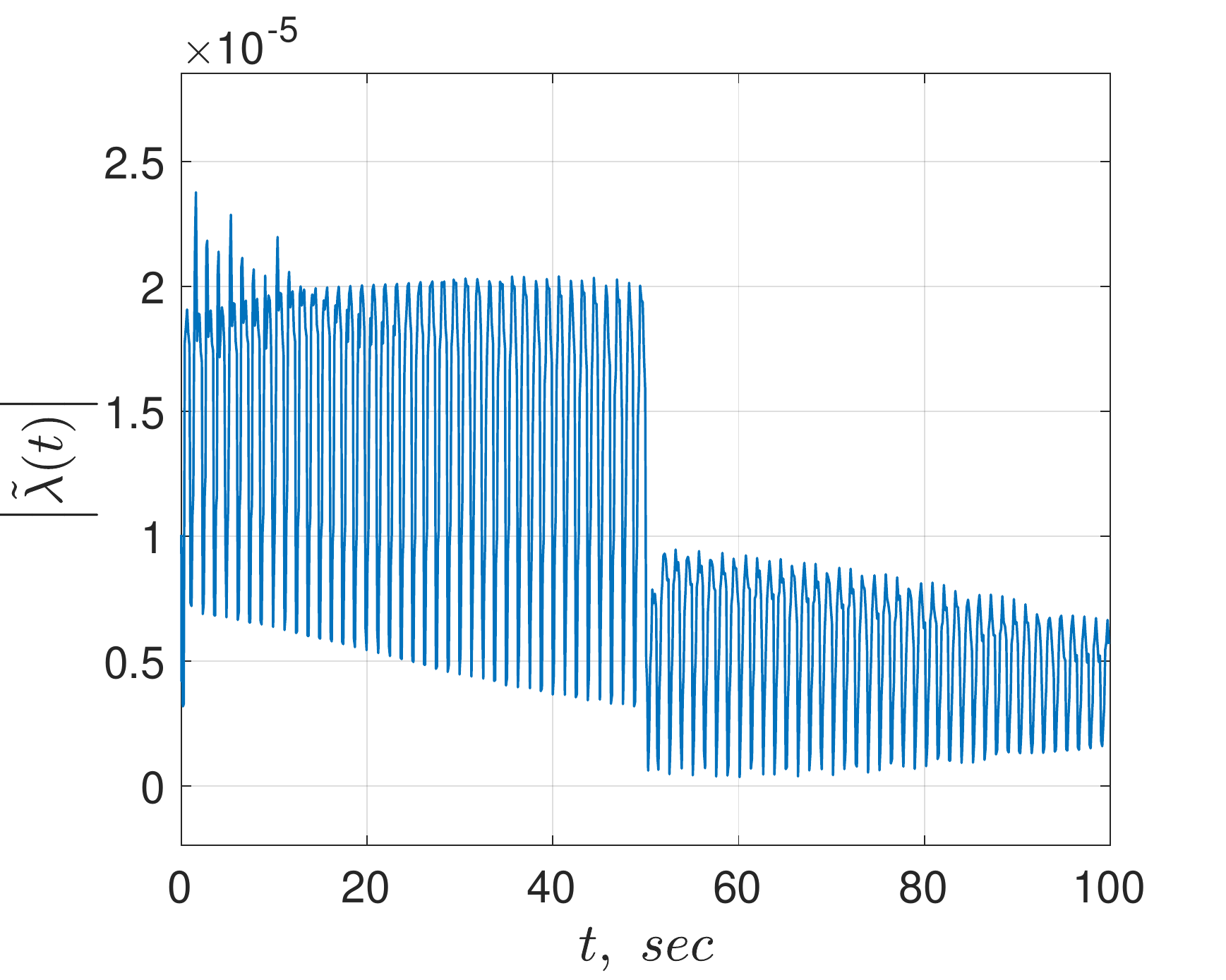}
\caption{\label{fig:robust} Norm of the robust flux observer error}
\end{figure}

\begin{figure}[htp]
\centering
\vspace{-3mm}
	\subcaptionbox{\label{fig:adaptive_flux} Norm of the flux estimation error}{\includegraphics[width=\linewidth]{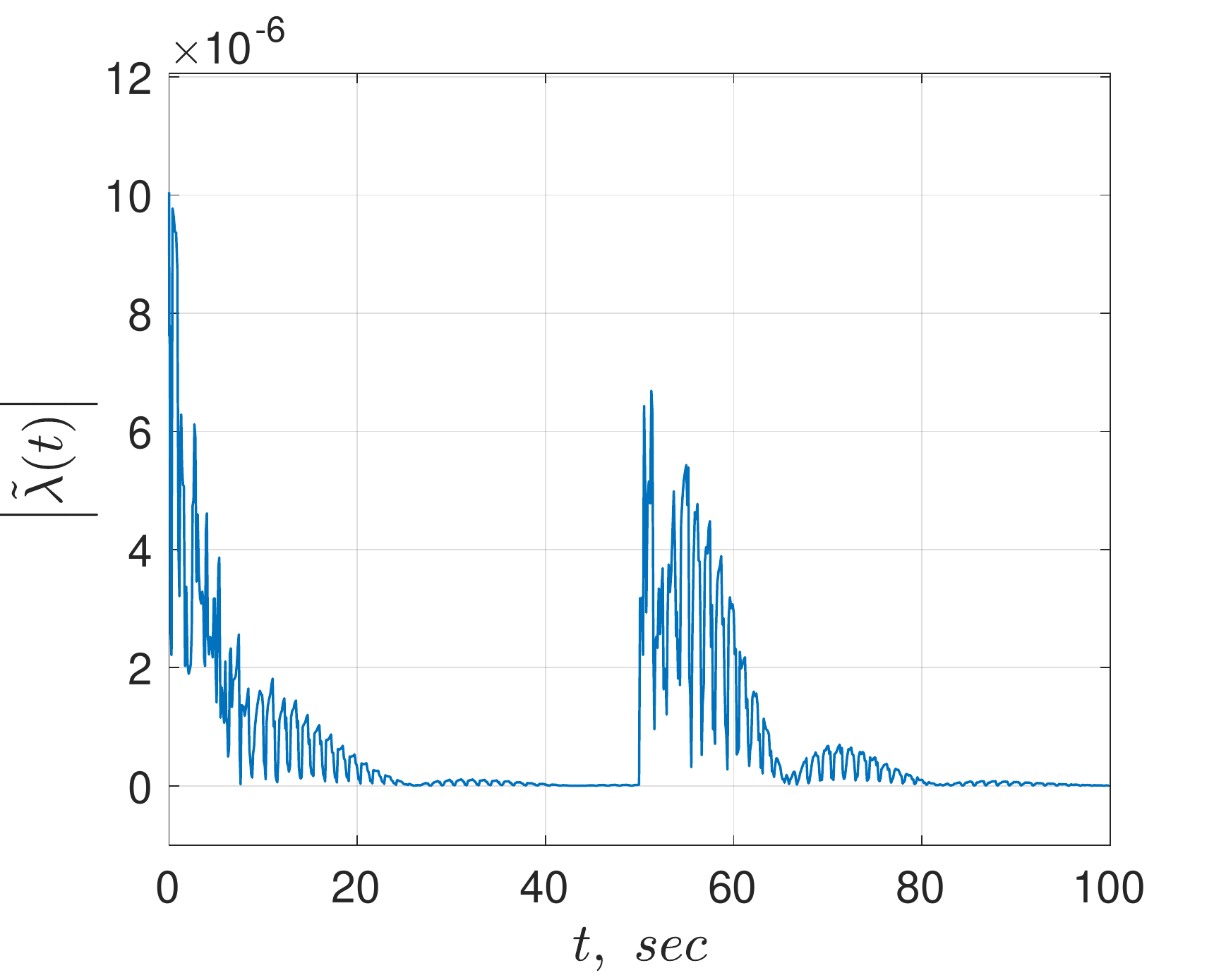}}
    \
	\subcaptionbox{\label{fig:adaptive_parameters} Norm of the parameters estimation error}{\includegraphics[width=\linewidth]{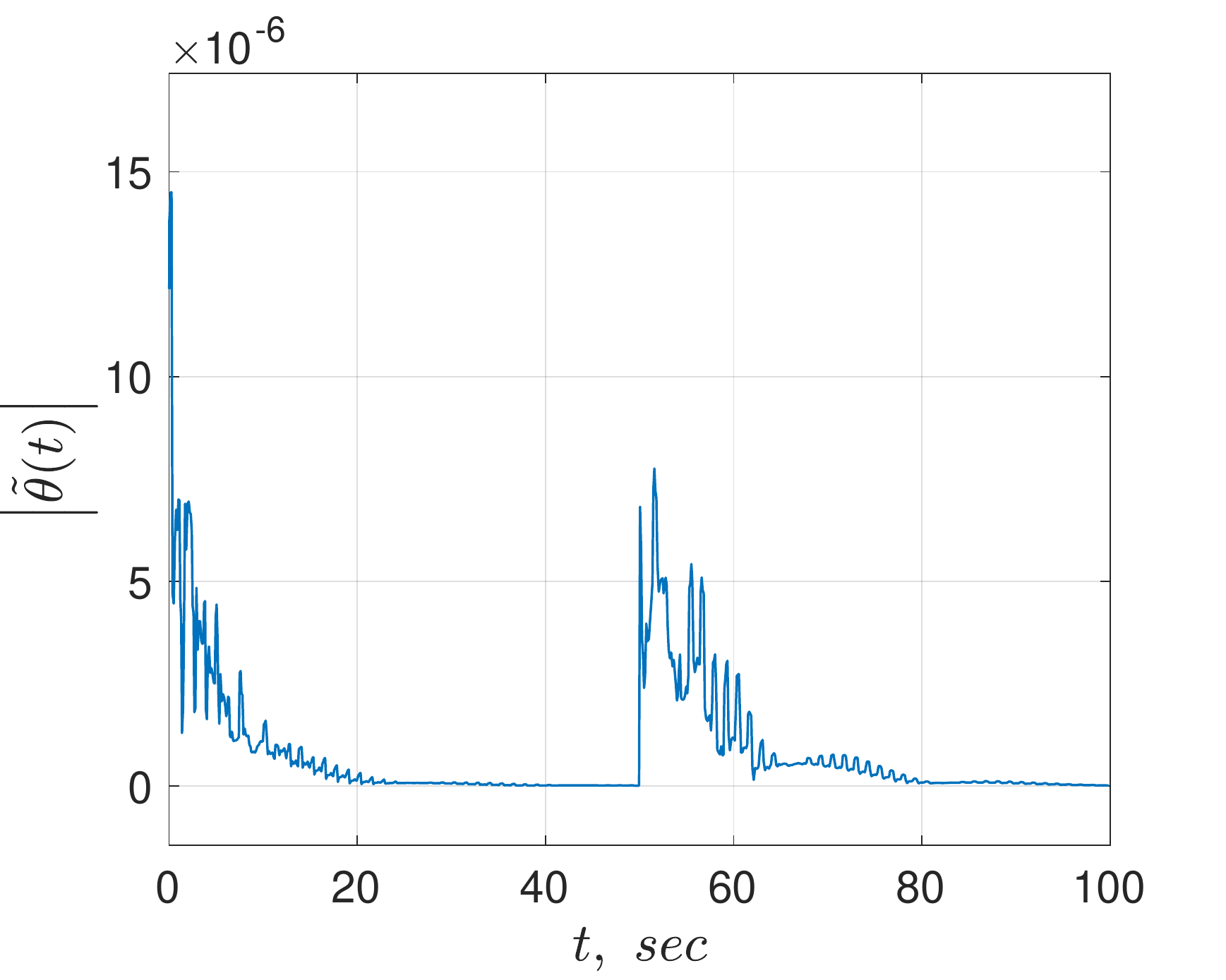}}
\vspace{-2mm}	
	\caption{\label{fig:adaptive} Adaptive flux observer}
\end{figure}

In Fig.~\ref{fig:robust} we illustrate the behaviour of the robust flux observer. The norm of the flux estimation error is bounded and, as indicated by the theory, the steady-state error depends on amplitude of the biases. Fig.~\ref{fig:adaptive} shows the results for the adaptive version of the flux observer. In this case, the norms of the flux and parameters estimation error converge to zero.

\section{Conclusions and Future Work}
\lab{sec7}
%
A solution to the robust flux observer problem posed in Subsection \ref{subsec22} for a class of electromechanical systems has been proposed. The class is identified by the, admittedly cryptic, Assumption \ref{ass1}. The main result of the paper is Lemma \ref{lem1} where it is shown that, via filtering and some nonlinear algebraic operations, it is possible to define the linear regression model \eqref{linear0} for this class of systems. Equipped with this regression model it is straightforward to apply a classical gradient design technique to estimate the flux and the unknown parameters.\footnote{Other alternatives, besides gradient-descent, are of course possible, for instance least squares or extended Kalman filters. But, in the present context, there are no clear advantages from the use of these more complicated schemes.}

Several open questions are currently being investigated.
\begite
\item A critical assumption in the paper is the PE condition that, as indicated in Remark \ref{rem6}, seems complicated to transfer to the external signals of \eqref{sys} or a particular operating regime. In the recent paper \cite{BARORT} relaxed conditions for global asymptotic (but not exponential) stability of the unperturbed error equations of interest in this paper have been reported. It is not clear at this point whether these weaker stability property can be used to prove asymptotic stability of the perturbed equations. 

\item A far reaching question is the identification of physical systems which verify the critical Assumption \ref{ass1}. This requires a detailed analysis of the physical modelling of magnetic field systems. A first, natural candidate to be investigated is the salient PMSM \cite{NAM}.

\item As shown in \cite{BOBetalaut} and \cite{BOBetalmaglev} it is possible to estimate $q$ and $\dot q$ from the knowledge of $\lambda$ for PMSM and Maglev systems, respectively. It is interesting to know under which conditions is it also possible for other classes of systems verifying Assumption \ref{ass1}. Obviously, this would imply an investigation on the mechanical energy function and the form of the toque of electrical origin $ \nabla_q H(\lambda,q) $.

\item Blondel-Parks transformability \cite{LIUetal,ORTbook} is a property of electrical machines that allows to remove the dependence on $q$ of the electrical energy function via a simple change of coordinates. This property is clearly similar to Assumption \ref{ass1} and their connection, though unclear at this point, worth investigated.  

\item We have considered in the paper only electromechanical systems with magnetic fields, but there are many practical examples where an electrical field is also present. How to tackle this class of systems is also a topic of current research.
\endite
%

%
\appendix
\section*{A. Proof of Lemma \ref{lem1}}
%

As suggested in \cite{BERPRA} to obtain a linear equation from the quadratic one \eqref{wim} we differentiate it to get 
\beal{
\lab{dif1}
0&=2\lambda^\top  Q_1\dot \lambda+\lambda^\top \dot y_a+\dot \lambda^\top  y_a+\dot \lambda^\top \theta_{y_a}+\dot y_b^\top \theta_{y_b}+\dot y_c\nonumber\\
&=\lambda^\top (2Q_1\dot \lambda+\dot y_a)+\dot \lambda^\top  y_a
+\dot \lambda^\top \theta_{y_a}+\dot y_b^\top \theta_{y_b}+\dot y_c
}
Substitution of \eqref{mea} into \eqref{dotlam} gives
\beal{
\lab{dotlam2}
\dot\lambda&=-R (i_m-\delta_i) + B (u_m-\delta_u),\nonumber\\
&=y_m+\theta_m,
}
where $y_m$ and $\theta_m$ are defined in \eqref{ym} and \eqref{the}, respectively. Clearly, \eqref{dotlam2} corresponds to the first equation of \eqref{linear0}.
Substitution of \eqref{dotlam2} into \eqref{dif1} yields
\begin{align}
\lab{dif2}
0 & = \lambda^\top (2Q_1y_m+2Q_1\theta_m+\dot y_a)+(y_m+\theta_m)^\top  y_a
	\nonumber \\
	& \quad + (y_m+\theta_m)^\top\theta_{y_a}+\dot y_b^\top  \theta_{y_b}+\dot y_c
\end{align}
Let us apply the operator ${\nu\over p+\nu}$ to \eqref{dif2}, where $p:={d\over dt}$.
\begin{align}
\lab{dif3}
0 & = {\nu\over p+\nu}\lambda^\top  (2Q_1y_m+2Q_1\theta_m+\dot y_a) + {\nu\over p+\nu}y_m^\top y_a
\nonumber \\
& \quad + \theta_m^\top  {\nu\over p+\nu}y_a + \theta_{y_a}^\top  {\nu\over p+\nu}y_m + \theta_{y_b}^\top  {p\nu\over p+\nu}y_b
\nonumber \\
& \quad + {\nu\over p+\nu}\theta_m^\top  \theta_{y_a}
+{p\nu\over p+\nu}y_c
+\et.
\end{align}

Now, let us recall the Swapping Lemma 3.6.5 of \cite{SASBOD}
$$
 {\nu\over p+\nu}(x^\top z)=z^\top( {\nu\over p+\nu}x)- {1\over p+\nu}[\dot z^\top( {\nu\over p+\nu}x)], 
 $$
which applied to the first right hand term in \eqref{dif3} yields
\begin{align}
&{\nu\over p+\nu}\lambda^\top  (2Q_1y_m+2Q_1\theta_m+\dot y_a)
\nonumber \\
& = \lambda^\top  {\nu\over p+\nu}\left[2Q_1y_m+2Q_1\theta_m+\dot y_a\right]
\nonumber \\
& \quad - {1\over p+\nu}\left[{(y_m+\theta_m)}^\top \times \right.
\nonumber \\
& \qquad \bigg. \times {\nu\over p+\nu} \left[2Q_1y_m+2Q_1\theta_m+\dot y_a\right]\bigg].
\nonumber
\end{align}
Combining this identity with the decomposition
$$
 {p\nu\over p+\nu}y_b=\nu y_b-{\nu^2\over p+\nu} y_b
$$
and grouping some terms yields
\begin{align}
0&=
\lambda^\top  \left[{\nu\over p+\nu}2Q_1y_m+{p\nu\over p+\nu} y_a\right]
\nonumber\\
& \quad + \lambda^\top \left[{\nu\over p+\nu}2Q_1\theta_m\right]
\nonumber\\
& \quad 
-{1\over p+\nu}\left[y_m^\top  {\nu\over p+\nu}[2Q_1y_m+\dot y_a]\right]
\nonumber\\
&\quad
-{1\over p+\nu}\left[\theta_m^\top  {\nu\over p+\nu}[2Q_1y_m+\dot y_a] \right.
\nonumber\\
&\left. \qquad + y_m^\top  {\nu\over p+\nu}2Q_1\theta_m
+\theta_m^\top  {\nu\over p+\nu}2Q_1\theta_m]
\right]
\nonumber\\
&\quad
+{\nu\over p+\nu}\left[y_m^\top  y_a-\nu y_c\right]
+\theta_m^\top  {\nu\over p+\nu}y_a
\nonumber \\
& \quad 
+\theta_{y_a}^\top {\nu\over p+\nu}y_m +\theta_{y_b}\left[\nu y_b-{\nu\over p+\nu}\nu y_b\right]
\nonumber\\
&\quad 
+{\nu\over p+\nu}\theta_m^\top \theta_{y_a} +\nu y_c+\et,
\nonumber
\end{align}
that may be represented in the form
\begin{align}
&{1\over p+\nu}\left[y_m^\top  \left[{\nu\over p+\nu}2Q_1y_m+{p\nu\over p+\nu} y_a\right]+\nu^2 y_c
\right. \nonumber \\
& \bigg. \qquad -\nu y_m^\top  y_a\bigg]-\nu y_c
\nonumber \\
& = \lambda^\top  \left[{\nu\over p+\nu}2Q_1y_m+{p\nu\over p+\nu} y_a \right] +\lambda^\top 2Q_1\theta_m
\nonumber \\
& \quad - {1\over p+\nu}\left[\theta_m^\top \left[{\nu\over p+\nu}2Q_1y_m
+ {p\nu\over p+\nu} y_a\right]
\right. \nonumber \\
& \bigg. \qquad + y_m^\top 2Q_1 \theta_m +\theta_m^\top 2Q_1\theta_m
\bigg] +\theta_m^\top  {\nu\over p+\nu}y_a
\nonumber\\
&\quad
+\theta_{y_a}^\top  {\nu\over p+\nu}y_m
+\theta_{y_b}\left[\nu y_b-{\nu^2\over p+\nu} y_b\right]
\nonumber\\
&\quad
+\theta_m^\top \theta_{y_a}+\et.
\nonumber
\end{align}

Invoking the first five equations of \eqref{fil} and realizing that we can use the following substitutions
\begin{align}
{\nu\over p+\nu}y_m&=\xi_1 \nonumber \\
{\nu\over p+\nu}2Q_1y_m+{p\nu\over p+\nu} y_a&=\xi_2+\nu y_a \nonumber \\
\nu y_b-{\nu^2\over p+\nu} y_b&=\nu y_b-\nu \xi_3 \nonumber
\end{align}
\begin{align}
& {1\over p+\nu}\left[ \left[{\nu\over p+\nu}2Q_1y_m+{p\nu \over p+\nu} y_a\right]+ 2Q_1y_m\right]
\nonumber \\
& \quad - {\nu\over p+\nu}y_a = \xi_4
\nonumber \\
& {1\over p+\nu}\left[y_m^\top  \left[{\nu\over p+\nu}2Q_1y_m+{p\nu\over p+\nu} y_a\right] \right.
\nonumber \\
& \quad \bigg. +\nu^2 y_c-\nu y_m^\top  y_a\bigg] = \xi_5,
\nonumber 
\end{align}
we obtain
\begin{align}
\xi_5-\nu y_c&=\lambda^\top [\xi_2+\nu y_a]+\lambda^\top \,2Q_1\theta_m-\theta_m^\top \xi_4
\nonumber \\
& \quad +\theta_{y_a}^\top \xi_1 +\theta_{y_b}^\top \nu[y_b-\xi_3]
\nonumber \\
& \quad + \theta_m^\top \left[\theta_{y_a}-{2Q_1 \over \nu} \theta_m\right]+\et.
\lab{dif6}
\end{align}

Now, apply the operator $p\over p+\nu$ to \eqref{dif6}
\begin{align}
& {p\over p+\nu}[\xi_5-\nu y_c]={p\over p+\nu}\lambda^\top [\xi_2+\nu y_a]
\nonumber \\
& \quad + {1\over p+\nu}\dot \lambda^\top \,2Q_1\theta_m-\theta_m^\top {p\over p+\nu}\xi_4+\theta_{y_a}^\top {p\over p+\nu}\xi_1
\nonumber\\
& \quad + \theta_{y_b}^\top {p\over p+\nu}\nu[y_b-\xi_3]
\nonumber\\
& \quad + {p\over p+\nu}\theta_m^\top \left[\theta_{y_a}-{2Q_1 \over \nu} \theta_m\right]+\et.
\lab{dif7}
\end{align}
Let us consider separately the term ${p\over p+\nu}\lambda^\top [\xi_2+\nu y_a]$
\begin{align}
& \frac{p}{p+\nu}\lambda^\top [\xi_2+\nu y_a]
\nonumber \\
& \quad =\frac{1}{p+\nu}\left[\dot \lambda^\top [\xi_2+\nu y_a]+\lambda^\top [\dot\xi_2+\nu \dot y_a]\right]
\nonumber \\
& \quad =\lambda^\top \left[{p\over p+\nu}[\xi_2+\nu y_a]\right]
\nonumber \\
& \qquad -{1\over p+\nu}\left[\dot \lambda^\top {p\over p+\nu}[\xi_2+\nu y_a] -\dot \lambda^\top [\xi_2+\nu y_a] \right]\nonumber\\
& \quad =\lambda^\top \left[{p\over p+\nu}[\xi_2+\nu y_a]\right]
\nonumber \\
& \qquad +{1 \over p+\nu}\left[\dot  \lambda^\top {\nu \over p+\nu}[\xi_2+\nu y_a] \right]
\nonumber\\
& \quad =\lambda^\top \left[\xi_2+\nu y_a-{\nu\over p+\nu}[\xi_2+\nu y_a]\right]
\nonumber \\
& \qquad +{1 \over p+\nu}\left[[y_m+\theta_m]^\top {\nu \over p+\nu}[\xi_2+\nu y_a] \right]\nonumber\\
& \quad =\lambda^\top \left[\xi_2+\nu y_a-{\nu \over p+\nu}[\xi_2+\nu y_a]\right]
\nonumber \\
& \qquad +{1 \over p+\nu} y_m^\top {\nu \over p+\nu} [\xi_2+\nu y_a]
\nonumber\\
& \qquad+\theta_m^\top {1 \over p+\nu}{\nu \over p+\nu}[\xi_2+\nu y_a],
\nonumber
\end{align}
and realize that
\beals{
{\nu \over p+\nu}[\xi_2+\nu y_a]=\nu\xi_4-\xi_2.
}
Hence, we can rewrite \eqref{dif7} as
\begin{align}
& {p\over p+\nu}[\xi_5-\nu y_c]=
\lambda^\top  [\xi_2+y_a-\nu\xi_4+\xi_2]
\nonumber\\
& \quad +{1 \over p+\nu}y_m^\top [\nu\xi_4-\xi_2] +\theta_m^\top{1 \over p+\nu}[\nu\xi_4-\xi_2]
\nonumber\\
&\quad+{1\over p+\nu}[y_m+\theta_m]^\top \, 2Q_1\theta_m-
\theta_m^\top \left[\xi_4-{\nu\over p+\nu}\xi_4\right]
\nonumber\\
&\quad+
\theta_{y_a}^\top \left[\xi_1-{\nu\over p+\nu}\xi_1\right]
\nonumber\\
&\quad +\theta_{y_b}^\top \left[\nu[y_b-\xi_3]-{\nu\over p+\nu}\nu[y_b-\xi_3]\right]+\et.
\nonumber
\end{align}

The proof is completed considering the four latter filters of \eqref{fil} to get
\begin{align}
& \xi_5-\nu y_c-\xi_9=
\lambda^\top  [2\xi_2+y_a-\nu\xi_4]
\nonumber \\
& \quad +
\theta_m^\top{2 \over p+\nu} [\nu\xi_4-\xi_2] + \theta_{y_a}^\top [\xi_1-\xi_7]
\nonumber\\
&\quad+\theta_{y_b}^\top \nu[y_b-\xi_3-\xi_8]+ {2 \over \nu}\theta_m^\top Q_1\theta_m+\et,
\nonumber
\end{align}
which is, precisely, the linear regression model \eqref{linear0} with the definitions \eqref{yphiphi} and  \eqref{the}.

\end{document}